\documentclass[12pt]{article}
\usepackage{enumerate}
\usepackage{amsfonts}
\usepackage{amsmath}
\usepackage{amssymb}
\usepackage{amsthm}

\def\H{\mathcal{H}}
\def\K{\mathcal{K}}
\def\P{\mathcal{P}}
\def\S{\mathfrak{S}}
\def\C{\mathfrak{C}}
\def\T{\mathfrak{T}}
\def\B{\mathfrak{B}}

\newcommand{\id}{\mathrm{Id}}
\newcommand{\Tr}{\mathrm{Tr}}

\newcounter{defin}  \newcounter{lemma}  \newcounter{theorem}
\newcounter{property} \newcounter{corol}  \newcounter{remark} \newcounter{example}

\newenvironment{lemma}{\par\refstepcounter{lemma}
     \textbf{Lemma \thelemma.} }{\rm\par}
\newenvironment{theorem}{\par\refstepcounter{theorem}
     \textbf{Theorem \thetheorem.}\ }{\rm\par}
\newenvironment{property}{\par\refstepcounter{property}
     \textbf{Proposition \theproperty.}\ }{\rm\par}
\newenvironment{corollary}{\par\refstepcounter{corol}
     \textbf{Corollary \thecorol.} }{\rm\par}
\newenvironment{definition}{\par\refstepcounter{defin}
     \textbf{Definition \thedefin.}\ }{\rm\par}
\newenvironment{remark}{\par\refstepcounter{remark}
     \textbf{Remark \theremark.}}{\rm\par}

\begin{document}
\title{Reversibility conditions for quantum channels and their applications.}
\author{M.E. Shirokov\footnote{email:msh@mi.ras.ru}\\
Steklov Mathematical Institute, RAS, Moscow}
\date{}
\maketitle 
\begin{abstract}
A necessary condition for reversibility (sufficiency) of a quantum
channel with respect to complete families of states with bounded
rank is obtained. A full description (up to isometrical equivalence)
of all quantum channels reversible with respect to orthogonal and
nonorthogonal complete families of pure states is given. Some
applications in quantum information theory are considered.

The main results can be formulated in terms of the operator algebras theory (as conditions for reversibility of 
channels between algebras of all bounded operators).
\end{abstract}
\vspace{5pt}
\tableofcontents

\section{Introduction}

Reversibility (sufficiency) of a quantum channel
$\Phi:\mathfrak{S}(\mathcal{H}_A)\rightarrow\mathfrak{S}(\mathcal{H}_B)$
with respect to a family $\S$ of states in
$\mathfrak{S}(\mathcal{H}_A)$ means existence of a quantum channel
$\Psi:\mathfrak{S}(\mathcal{H}_B)\rightarrow\mathfrak{S}(\mathcal{H}_A)$
such that $\Psi(\Phi(\rho))=\rho$ for all $\rho\in\S$.

The notion of reversibility of a channel naturally arises in
analysis of different general questions of quantum information
theory and quantum statistics \cite{H&Co,J&P,J-rev,O&Co,P-sqc,TEC}. For example,
the famous Petz's theorem states that an equality in the inequality
$$
H(\Phi(\rho)\hspace{1pt}\|\hspace{1pt}\Phi(\sigma))\leq H(\rho\hspace{1pt}\|\hspace{1pt}\sigma),\quad
\rho,\hspace{1pt}\sigma\in\mathfrak{S}(\mathcal{H}_A),
$$
expressing the fundamental monotonicity property of the quantum relative
entropy, holds if and only if the channel $\Phi$ is reversible with
respect to the states $\rho$ and $\sigma$.

It follows from this theorem that the Holevo quantity\footnote{The Holevo
quantity $\chi(\{\pi_i,\rho_i\})\doteq\sum_i\pi_i
H(\rho_i\|\hspace{1pt}\bar{\rho})$, where
$\bar{\rho}=\sum_i\pi_i\rho_i$, provides an upper bound for
accessible classical information which can be obtained by applying a
quantum measurement \cite{H-SCI,N&Ch}.} of an ensemble
$\{\pi_i,\rho_i\}$ of quantum states is preserved under action of a
quantum channel $\Phi$, i.e.
$$
\chi(\{\pi_i,\Phi(\rho_i)\})=\chi(\{\pi_i,\rho_i\}),
$$
if and only if the channel $\Phi$ is reversible with respect to the
family $\{\rho_i\}$. Further analysis shows that preserving
conditions for many others important characteristics under action of
a quantum channel are also reduced to the reversibility condition
\cite{H&Co,J-rev}. In \cite{TEC} it is shown that a criterion for an
equality between the constrained Holevo capacity and the quantum
mutual information of a quantum channel $\Phi$ can be formulated in
terms of reversibility of the complementary channel $\widehat{\Phi}$
with respect to particular families of pure states.\smallskip

In this paper we study conditions for reversibility of a quantum
channel by using the notion of a complementary channel, whose
essential role in analysis of different problems of quantum
information theory was shown recently \cite{H-c-c,KMNR}. By using Petz's
theorem we prove that reversibility  of a quantum channel with
respect to complete families of states with rank $\leq r$ implies
that the complementary channel has the Kraus representation
consisting of operators with rank $\leq r$ (Theorem \ref{main}). In
the case of families of pure states (states with rank $=1$) this
result leads to simple criterion of reversibility, which gives a
full description (up to isometrical equivalence) of all quantum
channels reversible with respect to given orthogonal and
nonorthogonal complete families of pure states (Proposition \ref{orth-f}, Theorem
\ref{non-orth-2}).

Some applications of the obtained results in quantum information
theory are considered in the last part of the paper (Theorem
\ref{main++} and its corollaries).

\section{Preliminaries} Let $\H$ be either a finite dimensional or
separable Hilbert space, $\B(\H)$ and
$\mathfrak{T}( \mathcal{H})$ -- the Banach spaces of all bounded
operators in $\mathcal{H}$ and of all trace-class operators in
$\H$  correspondingly, $\S(\H)$ -- the closed convex subset
of $\mathfrak{T}( \H)$ consisting of positive operators
with unit trace called \emph{states} \cite{H-SCI,N&Ch}.\smallskip

Denote by $I_{\mathcal{H}}$ and $\mathrm{Id}_{\mathcal{H}}$ the
unit operator in a Hilbert space $\mathcal{H}$ and the identity
transformation of the Banach space $\mathfrak{T}(\mathcal{H})$
correspondingly.\smallskip

A family $\{|\psi_{i}\rangle\}$ of vectors in a Hilbert space $\H$ is
called \emph{overcomplete} if
$$
\sum_{i}|\psi_{i}\rangle\langle \psi_{i}|=I_{\H}.
$$

Let $H(\rho)$ and $H(\rho\|\sigma)$ be respectively the von Neumann
entropy of the state $\rho$ and the quantum relative entropy of the
states $\rho$ and $\sigma$~\cite{H-SCI,N&Ch,O&P}.\smallskip

A linear completely positive trace preserving map
$\Phi:\mathfrak{T}(\mathcal{H}_A)\rightarrow\mathfrak{T}(\mathcal{H}_B)$
is called  \emph{quantum channel} \cite{H-SCI,N&Ch}.
\smallskip

For a given channel
$\Phi:\mathfrak{T}(\mathcal{H}_A)\rightarrow\mathfrak{T}(\mathcal{H}_B)$
the Stinespring theorem implies existence of a Hilbert space
$\mathcal{H}_E$ and of an isometry
$V:\mathcal{H}_A\rightarrow\mathcal{H}_B\otimes\mathcal{H}_E$ such
that
\begin{equation}\label{Stinespring-rep}
\Phi(\rho)=\mathrm{Tr}_{\mathcal{H}_E}V\rho V^{*},\quad
\rho\in\mathfrak{T}(\mathcal{H}_A).
\end{equation}
A quantum  channel
\begin{equation}\label{c-channel}
\mathfrak{T}(\mathcal{H}_A)\ni
\rho\mapsto\widehat{\Phi}(\rho)=\mathrm{Tr}_{\mathcal{H}_B}V\rho V^{*}\in\mathfrak{T}(\mathcal{H}_E)
\end{equation}
is called \emph{complementary} to the channel $\Phi$
\cite{H-c-c}.\footnote{The quantum channel $\widehat{\Phi}$ is also
called \emph{conjugate} to the channel $\Phi$ \cite{KMNR}.} The
complementary channel is defined uniquely in the following sense: if
$\widehat{\Phi}':\mathfrak{T}(\mathcal{H}_A)\rightarrow\mathfrak{T}(\mathcal{H}_{E'})$
is a channel defined by (\ref{c-channel}) via the Stinespring
isometry
$V':\mathcal{H}_A\rightarrow\mathcal{H}_B\otimes\mathcal{H}_{E'}$
then the channels $\widehat{\Phi}$ and $\widehat{\Phi}'$ are
isometrically equivalent in the sense of the following definition
\cite{H-c-c}.\smallskip

\begin{definition}\label{isom-eq}
Channels
$\Phi:\mathfrak{T}(\mathcal{H}_A)\rightarrow\mathfrak{T}(\mathcal{H}_B)$
and
$\,\Phi':\mathfrak{T}(\mathcal{H}_{A})\rightarrow\mathfrak{T}(\mathcal{H}_B')\,$
are \emph{isometrically equivalent} if there exists a partial
isometry $W:\mathcal{H}_B\rightarrow\mathcal{H}_{B'}$ such that
\begin{equation}\label{c-isom}
\Phi'(\rho)=W\Phi(\rho)W^*,\quad\Phi(\rho)=W^*\Phi'(\rho)W,\quad \rho\in \T(\H_A).
\end{equation}
\end{definition}

The notion of isometrical equivalence is very close to the notion of
unitary equivalence. Indeed, the isometrical equivalence of the
channels $\Phi$ and $\Phi'$ means unitary equivalence of these
channels with the output spaces $\H_B$ and $\H_{B'}$ replaced by
their subspaces
$\H^{\Phi}_B=\bigvee_{\rho\in\S(\H_A)}\mathrm{supp}\Phi(\rho)$ and
$\H^{\Phi'}_{B'}=\bigvee_{\rho\in\S(\H_A)}\mathrm{supp}\Phi'(\rho)$.\footnote{We
denote by $\mathrm{supp}\rho$ the support of a state $\rho$ (the
subspace $(\mathrm{ker}\rho)^{\perp}$).} We use the notion of
isometrical equivalence, since dealing with a given representation
of a quantum channel $\Phi$ it not easy in general to determine the
corresponding subspace $\H^{\Phi}_B$. \smallskip

The Stinespring representation (\ref{Stinespring-rep}) is called
\emph{minimal} if the subspace
$$
\mathcal{M}=\left\{\,(X\otimes
I_E)V|\varphi\rangle\;|\;\varphi\in\H_A,\, X\in\B(\H_B)\,\right\}
$$
is dense in $\mathcal{H}_B\otimes\mathcal{H}_E$. The complementary
channel $\widehat{\Phi}$ defined by (\ref{c-channel}) via the
minimal Stinespring representation has the following property:
\begin{equation}\label{full-rank}
\widehat{\Phi}(\rho)\;\text{is a full rank state
in}\;\S(\H_E)\;\text{for any full rank
state}\;\rho\;\textrm{in}\;\S(\H_A).
\end{equation}

The Stinespring representation (\ref{Stinespring-rep}) generates the
Kraus representation
\begin{equation}\label{Kraus-rep}
\Phi(\rho)=\sum_{k}V_{k}\rho V^{*}_{k},\quad
\rho\in\,\mathfrak{T}(\mathcal{H}),
\end{equation}
where $\{V_{k}\}$ is the set of bounded linear operators from
$\mathcal{H}_A$ into $\mathcal{H}_B$ such that
$\sum_{k}V^{*}_{k}V_{k}=I_{\H_A}$ defined by the relation
$$
\langle\varphi|V_k\psi\rangle=\langle\varphi\otimes
k|V\psi\rangle,\quad\varphi\in\H_B,\psi\in\H_A,
$$
where $\{|k\rangle\}$ is a particular orthonormal basis in the space
$\mathcal{H}_E$. The corresponding complementary channel is
expressed as follows
\begin{equation}\label{Kraus-rep-c}
\widehat{\Phi}(\rho)=\sum_{k,l}\mathrm{Tr}\left[V_{k}\rho V_{l}^{*}\right]|k\rangle\langle
l|,\quad \rho\in\,\mathfrak{T}(\mathcal{H}_A).
\end{equation}

The following class of quantum channels  plays an essential role in
this paper \cite{H-SCI,N&Ch}.
\smallskip
\begin{definition}\label{c-q-def}
A channel $\Phi:\T(\H_A)\rightarrow\T(\H_B)$ is called
\emph{classical-quantum} (briefly, \emph{c-q channel}) if it has the
following representation
\begin{equation}\label{c-q-rep}
\Phi(\rho)=\sum_{k=1}^{\dim\H_A}\langle
k|\rho|k\rangle\sigma_k,\quad \rho\in \T(\H_A),
\end{equation}
where $\{|k\rangle\}$ is an orthonormal basis in $\H_A$ and
$\{\sigma_k\}$ is a collection of states in $\S(\H_B)$.
\end{definition}\smallskip

C-q channel (\ref{c-q-rep}) for which $\sigma_k=\sigma$ for all
$k$ is a \emph{completely depolarizing} channel
$\Phi(\rho)=\sigma\Tr\rho$, where $\sigma$ is a given state in
$\S(\H_B)$.
\medskip

The Schmidt rank of a pure state $\omega$ in $\S(\H\otimes\K)$ can
be defined as the operator rank of the isomorphic
states $\Tr_{\K}\omega$ and $\Tr_{\H}\omega$ \cite{T&H}.\smallskip

The Schmidt class $\,\S_r$ of order $r\in\mathbb{N}\,$ is the
minimal convex closed subset of $\S(\H\otimes\K)$ containing all
pure states with the Schmidt rank $\leq r$, i.e. $\S_r$ is the
convex closure of these pure states \cite{T&H,Sh-18}.\footnote{In
finite dimensions the convex closure coincides with the convex hull
by the Caratheodory theorem, but in infinite dimensions even the set
of all \emph{countable} convex mixtures of pure states with the
Schmidt rank $\leq r$ is a proper subset of $\S_r$ for each $r$
\cite{Sh-18}.} In this notations $\S_1$ is the set of all separable
(non-entangled) states in $\S(\H\otimes\K)$.\smallskip

A channel $\Phi:\T(\H_A)\rightarrow\T(\H_B)$ is called \emph{entanglement-breaking} if for an
arbitrary Hilbert space $\K$ the state $\Phi\otimes\id_{\K}(\omega)$
is separable for any state $\omega\in\S(\H_A\otimes\K)$
\cite{e-b-ch}. This notion is generalized in \cite{p-e-b-ch} as
follows.\smallskip\pagebreak

\begin{definition}\label{p-e-b-ch-d}
A channel $\Phi:\T(\H_A)\rightarrow\T(\H_B)$ is called
$r$-\emph{partially entanglement-breaking }(briefly $r$-PEB) if for
an arbitrary Hilbert space $\K$  the state
$\Phi\otimes\id_{\K}(\omega)$ belongs to the Schmidt class
$\S_r\subset\S(\H_B\otimes\K)$ for any state
$\omega\in\S(\H_A\otimes\K)$.\smallskip
\end{definition}

In this notations entanglement-breaking channels are $1$-PEB
channels. Properties of $r$-PEB channels in finite dimensions are
studied in \cite{p-e-b-ch}, where it is proved, in particular, that
the class of $r$-PEB channels coincides with the class of channels
having  Kraus representation (\ref{Kraus-rep}) such that
$\mathrm{rank}V_k\leq r$ for all $k$. But in infinite dimensions the
first class is essentially wider than the second one, moreover, for
each $r$ there exist $r$-PEB channels  such that all operators in
any their Kraus representations have infinite rank
\cite{Sh-18}.\medskip

Following \cite{J-rev,O&Co} introduce the basic notion of this
paper.
\smallskip
\begin{definition}\label{rev-def}
A channel $\Phi:\T(\H_A)\rightarrow\T(\H_B)$ is \emph{reversible}
with respect to a family $\S\subseteq\S(\H_A)$ if there exists a
channel $\,\Psi:\T(\H_B)\rightarrow\T(\H_A)$ such that
$\,\rho=\Psi\circ\Phi(\rho)\,$ for all $\,\rho\in\S$.\footnote{This
property is also called sufficiency of the channel $\Phi$ with
respect to the family $\S$ \cite{J&P,P-sqc}.}
\end{definition}\smallskip

The channel $\Psi$ will be called \emph{reversing channel}.
\smallskip

Note that  reversibility is a common property for
isometrically equivalent channels.\smallskip
\begin{lemma}\label{isom-eq-l}
\emph{Let $\,\Phi:\T(\H_A)\rightarrow\T(\H_B)$ and
$\,\Phi':\T(\H_A)\rightarrow\T(\H_{B'})$ be quantum channels
isometrically equivalent in the sense of Def.\ref{isom-eq}. If the
channel $\,\Phi$ is reversible with respect to a family
$\,\S\subseteq\S(\H_A)$ then the channel $\,\Phi'$ is reversible
with respect to this family $\,\S$ and vice versa.}
\end{lemma}\smallskip
\textbf{Proof.} Let $\Psi$ be a reversing channel for the channel
$\Phi$, i.e. $\Psi\circ\Phi(\rho)=\rho$ for all $\rho\in\S$.
Consider the channel $\Theta(\cdot)=W^*(\cdot)W+\sigma\Tr(I_{\H_{B'}}-WW^*)(\cdot)$ from
$\S(\H_{B'})$ into $\S(\H_{B})$, where $W$ is the partial isometry
in (\ref{c-isom}) and $\sigma$ is a given state in $\S(\H_B)$. It is easy to see that
$\Psi\circ\Theta$ is a reversing channel for the channel $\Phi'$.
$\square$
\smallskip

Petz's theorem gives criterion of reversibility of a channel with respect to families of two states.
It will be used in this paper in the following
reduced form.\smallskip

\begin{theorem}\label{P-th}\cite{P-sqc}
\emph{Let $\,\Phi:\T(\H_A)\rightarrow\T(\H_B)$ be a  quantum
channel, $\rho$ and $\sigma$ states in $\S(\H_A)$ such that
$H(\rho\hspace{1pt}\|\hspace{1pt}\sigma)<+\infty$. Let $\,\Theta_{\sigma}:\T(\H_B)\rightarrow\T(\H_A)$ be
the predual channel to the linear completely positive unital map
$$
\Theta^*_{\sigma}(\cdot)=A\Phi\left(B(\cdot)B\right)A,\quad
A=[\Phi(\sigma)]^{-1/2},\; B=[\sigma]^{1/2},
$$
from $\B(\H_A)$ into $\B(\H_B)$. The following statements are equivalent:}
\begin{enumerate}[(i)]
    \item $H(\Phi(\rho)\|\hspace{1pt}\Phi(\sigma))=H(\rho\hspace{1pt}\|\hspace{1pt}\sigma)$;
    \item \emph{the channel $\,\Phi$ is reversible with respect the states $\rho$ and $\sigma$};
    \item $\rho=\Theta_{\sigma}(\Phi(\rho))$.
\end{enumerate}
\end{theorem}

Note that $\sigma=\Theta_{\sigma}(\Phi(\sigma))$ by definition of the channel $\Theta_{\sigma}$.

This theorem is proved in \cite{P-sqc} in general von Neumann
algebras setting for normal faithful states, i.e. for full rank
states  $\rho$ and $\sigma$ in our terminology. Since the condition
$H(\rho\hspace{1pt}\|\hspace{1pt}\sigma)<+\infty$ implies
$\,\mathrm{supp}\rho\subseteq\mathrm{supp}\sigma$, we always may
assume that $\sigma$ is a full rank state. A possible generalization
to the case $\mathrm{supp}\rho\neq\H_A$ is presented in Appendix 5.1
(in finite dimensions it follows from the Theorem in
\cite[Sect.5.1]{H&Co}).\smallskip

\begin{definition}\label{comp}
A family $\S$ of states in $\S(\H)$ is called \emph{complete} if for any
nonzero operator $A$ in $\B_{+}(\H)$ there exists a state
$\rho\in\S$ such that $\Tr A\rho>0$.
\end{definition}\smallskip

A family $\{|\varphi_{\lambda}\rangle\langle\varphi_{\lambda}|\}_{\lambda\in\Lambda}$ of pure states in
$\S(\H)$ is complete if and only if the linear hull of the family
$\{|\varphi_{\lambda}\rangle\}_{\lambda\in\Lambda}$ is dense in $\H$. By Lemma 2 in \cite{J&P}
an arbitrary complete family of states in $\S(\H)$ contains a
countable complete subfamily.\smallskip

Petz's theorem implies the following criterion for reversibility of
a channel with respect to countable complete families of states.
\smallskip

\begin{theorem}\label{PJ-th}\cite{J&P}
\emph{A quantum channel $\,\Phi:\T(\H_A)\rightarrow\T(\H_B)$ is
reversible with respect to a complete countable family
$\,\{\rho_i\}$ of states in $\,\S(\H_A)$ if and only if
$\,\rho_i=\Theta_{\bar{\rho}}(\Phi(\rho_i))$ for all $\,i$, where
$\bar{\rho}=\sum_i\pi_i\rho_i$ and $\,\{\pi_i\}$ is any
non-degenerate probability distribution.}
\end{theorem}
\smallskip

\section{Conditions for reversibility of a channel}

\subsection{Families of states with bounded rank}

For a given channel $\,\Phi:\T(\H_A)\rightarrow\T(\H_B)$ let
$\H^{\Phi}_B=\bigvee_{\rho\in\S(\H_A)}\mathrm{supp}\Phi(\rho)$ and
$$
m(\Phi)=\left\{\begin{array}{lc}
        \dim\ker\Phi^* & \textup{if}\;\, \H^{\Phi}_B=\H_B \\
        \dim\ker\Phi^*|_{\B(\H^{\Phi}_B)} & \textup{if}\;\, \H^{\Phi}_B\neq\H_B
        \end{array}\right.,
$$
where $\Phi^*|_{\B(\H^{\Phi}_B)}$ is the restriction of the dual map
$\,\Phi^*:\B(\H_B)\rightarrow\B(\H_A)$ to the subspace
$\B(\H^{\Phi}_B)$ of $\B(\H_B)$. It is clear that
$m(\Phi)=\min_{\Psi\sim\Phi}\dim\ker\Psi^*$, where the minimum is
over all channels $\Psi$ isometrically equivalent to the channel
$\Phi$ in the sense of Def.\ref{isom-eq}.\smallskip

Petz's theorem implies the following necessary condition for
reversibility of a quantum channel with respect to families of
states with bounded rank expressed in terms of the complementary
channel.\smallskip

\begin{theorem}\label{main} \emph{Let $\,\S=\{\rho_i\}^n_{i=1},$ $n\leq+\infty$,
be a complete family of states in $\,\S(\H_A)$ such that
$\,\mathrm{rank}\rho_i\leq r\,$ for all $\,i$.} \emph{If a channel
$\,\Phi:\T(\H_A)\rightarrow\T(\H_B)$ is reversible with respect to
the family $\,\S$ then its complementary channel $\,\widehat{\Phi}$
has Kraus representation (\ref{Kraus-rep}) consisting of $\,\leq
n\times\min\{m(\Phi)+r^2, \dim\H^{\Phi}_B\}$ summands such that
$\;\mathrm{rank}V_k\leq r$ for all $\,k$ and hence
$\,\widehat{\Phi}$ is a $r$-partially entanglement-breaking channel
(Def.\ref{p-e-b-ch-d}).}
\smallskip

\emph{If the above hypothesis holds with $\,r=1\,$, i.e.
$\rho_i=|\varphi_i\rangle\langle\varphi_i|$ for all $\,i$, then
\begin{equation}\label{c-ch-rep}
\widehat{\Phi}(\rho)=\sum_{i=1}^n\langle\phi_i|\rho|\phi_i\rangle\sum_{k=1}^m
|\psi_{ik}\rangle\langle\psi_{ik}|,
\end{equation}
where $\,m=\min\{m(\Phi)+1, \dim\H^{\Phi}_B\}$,
$\,\{|\phi_i\rangle\}^n_{i=1}$ is an overcomplete system of vectors in
$\H_A$ defined by means of an arbitrary non-degenerate probability
distribution $\,\{\pi_i\}^n_{i=1}$ as follows
\begin{equation}\label{v-phi-rep}
|\phi_i\rangle=\sqrt{\pi_i\bar{\rho}_{\pi}^{\,-1}}|\varphi_i\rangle,\quad
\bar{\rho}_{\pi}=\sum^n_{i=1}\pi_i|\varphi_i\rangle\langle\varphi_i|,
\end{equation}
and $\,\{|\psi_{ik}\rangle\}$ is a collection of vectors in a Hilbert
space $\H_E$ such that\break $\sum_{k=1}^m \|\psi_{ik}\|^2=1$ and
$\,\langle\psi_{il}|\psi_{ik}\rangle=0$ for all $\,k\neq l$ for each
$\,i=\overline{1,n}$.}
\end{theorem}
\medskip

The first assertion of Theorem \ref{main} means that the channel
$\widehat{\Phi}$ has the following property:  for an arbitrary
Hilbert space $\K$  and any state $\omega$ in $\S(\H_A\otimes\K)$
the state $\widehat{\Phi}\otimes\id_{\K}(\omega)$ is a
\emph{countably decomposable} state in the Schmidt class
$\S_r\subset\S(\H_B\otimes\K)$, i.e. it can be represented as a
countable convex mixture of pure states having the Schmidt rank
$\leq r$  (there exist states in $\S_r$ which are not countably
decomposable \cite{Sh-18}).\smallskip

The second assertion of Theorem \ref{main} implies the criteria of reversibility of a quantum channel
with respect to orthogonal families of pure states
considered in the next subsection (Proposition \ref{orth-f} and Corollary \ref{orth-f-c}).\smallskip

\textbf{Proof.}  Let $\,\widehat{\Phi}(\rho)=\sum_{k=1}^{d}V_k\rho
V_k^*$, $\,d\leq+\infty$, be the Kraus representation of the channel
$\widehat{\Phi}:\T(\H_A)\rightarrow\T(\H_E)$ obtained via its
minimal Stinespring representation with the isometry
$V:\H_A\rightarrow\H_E\otimes\H_{C}$ (see Section 2). The
complementary channel $\Psi=\widehat{\widehat{\Phi}}$ to the channel
$\widehat{\Phi}$ defined via this representation is expressed as
follows
$$
\T(\H_A)\ni\rho\mapsto\Psi(\rho)=\sum_{k,l=1}^d \Tr V_k\rho
V_l^*|k\rangle\langle l|\in\T(\H_{C}),
$$
where $\,\{|\,k\rangle\}_{k=1}^d$ is an orthonormal basis in the
$d$-dimensional Hilbert space $\H_C$.

Since $\Psi=\widehat{\widehat{\Phi}}$, the channels $\Phi$ and
$\Psi$ are isometrically equivalent. By Lemma \ref{isom-eq} below
the channel $\Psi$ is reversible with respect to the set
$\{\rho_i\}$.\smallskip

Let $\{\pi_i\}^n_{i=1}$ be an arbitrary non-degenerate probability
distribution and $\bar{\rho}=\sum^n_{i=1}\pi_i\rho_i$. By property (\ref{full-rank})
$\Psi(\bar{\rho})$ is a full rank state in $\S(\H_C)$. By Theorem
\ref{PJ-th} the reversibility condition implies $A_i=\Psi^*(B_i)$
for all $i$, where
$A_i=\pi_i(\bar{\rho})^{-1/2}\rho_i(\bar{\rho})^{-1/2}$ and
$B_i=\pi_i(\Psi(\bar{\rho}))^{-1/2}\Psi(\rho_i)(\Psi(\bar{\rho}))^{-1/2}$
are positive operators in $\B(\H_A)$ and in $\B(\H_{C})$
correspondingly.

Note that
$$
\Psi^*(C)=\sum_{k,l=1}^d \langle l |C|k\rangle V_l^*V_k,\quad
C\in\B(\H_C).
$$

Since $A_i=\Psi^*(B_i)$ is an operator of rank $\leq r$, Lemma
\ref{s-rep} below implies $B_i=\sum_{j=1}^m|\psi_{ij}\rangle\langle
\psi_{ij}|$, where $m=\min\{\dim\ker\Psi^*+r^2,\, \dim\H_C\}$ and
$\,\{|\psi_{ij}\rangle\}_j$ is a set of vectors in $\H_C$, for each
$i$.

Since $\Psi(\bar{\rho})$ is a full rank state in $\S(\H_C)$, we have
\begin{equation*}
\sum_{i=1}^n\sum_{j=1}^m|\psi_{ij}\rangle\langle
\psi_{ij}|=\sum_{i=1}^n B_i=I_{\H_C}.
\end{equation*} By Lemma
\ref{new-kraus-rep} below
\begin{equation}\label{c-ch-rep+}
\widehat{\Phi}(\rho)=\sum_{i=1}^n\sum_{j=1}^m W_{ij}\rho W_{ij}^*,
\end{equation}
where $W_{ij}=\sum^d_{k=1}\langle\psi_{ij}|k\rangle V_k$.\smallskip

Since $A_i=\Psi^*(\sum_{j=1}^m|\psi_{ij}\rangle\langle
\psi_{ij}|)\,$ is an operator of rank $\leq r$ for each $i$ and
\begin{equation}\label{n-eq}
\Psi^*(|\psi_{ij}\rangle\langle \psi_{ij}|)=\sum_{k,l=1}^d \langle l
|\psi_{ij}\rangle\langle\psi_{ij}|k\rangle
V_l^*V_k=\,W_{ij}^*W_{ij},
\end{equation}
the family $\{W_{ij}\}$ consists of operators of rank $\leq r$. To
complete the proof of the first part of the theorem it suffices to
note that $\dim\H_C=\dim\H^{\Phi}_B$ and  $\dim\ker\Psi^*=m(\Phi)$,
since the partial isometry expressing the isometrical equivalence of
the channels $\Phi$ and $\Psi$ is an isometrical embedding of $\H_C$
into $\H_B$ (due to full rank of the state $\Psi(\bar{\rho})\in\S(\H_C)$).
\smallskip

If $\rho_i=|\varphi_i\rangle\langle\varphi_i|$ for each $i$ then
$A_i=|\phi_{i}\rangle\langle\phi_{i}|$, where the vector
$|\phi_{i}\rangle$ is defined by (\ref{v-phi-rep}), and (\ref{n-eq})
implies
$$
|\phi_{i}\rangle\langle\phi_{i}|=\sum_{j=1}^m\Psi^*(|\psi_{ij}\rangle\langle
\psi_{ij}|)=\sum_{j=1}^m W_{ij}^*W_{ij}.
$$
Hence $W_{ij}=|\eta_{ij}\rangle\langle\phi_{i}|$ for all $i$ and
$j$, where $\{|\eta_{ij}\rangle\}$ is a set of vectors in $\H_E$
such that $\sum_{j=1}^m \|\eta_{ij}\|^2=1$ for each
$i=\overline{1,n}$.\smallskip

It follows from (\ref{c-ch-rep+}) that
$$
\widehat{\Phi}(\rho)=\sum_{i=1}^n\langle\phi_i|\rho|\phi_i\rangle\sum_{j=1}^{m}
|\eta_{ij}\rangle\langle\eta_{ij}|,\quad \rho\in\T(\H_A),
$$
By using spectral decomposition of the states $\sum_{j=1}^{m}
|\eta_{ij}\rangle\langle\eta_{ij}|$, $i=\overline{1,n}$, we obtain
representation (\ref{c-ch-rep}). $\square$
\smallskip

\begin{lemma}\label{s-rep}
\emph{Let $\,\Phi:\T(\H_A)\rightarrow\T(\H_B)$ be a quantum
channel.} \emph{If $B$ is a positive operator in $\B(\H_B)$ such
that $\,\mathrm{rank}\Phi^*(B)=r<+\infty$ then
$B=\sum_{j=1}^{m}|\psi_j\rangle\langle\psi_j|$, where
$\,m=\min\left\{\dim\ker\Phi^*+r^2,\, \dim\H_B\right\}$ and
$\,\{|\psi_j\rangle\}_{j=1}^{m}$ is a set of vectors in $\H_B$.}
\end{lemma}\smallskip

\textbf{Proof.} Note first that for an arbitrary orthonormal basis
$\{|j\rangle\}$ in $\H_B$ we have
$B=\sum_{j=1}^{\dim\H_B}|\psi_j\rangle\langle\psi_j|$, where
$|\psi_j\rangle=B^{1/2}|j\rangle$. So, the assertion of the lemma is
nontrivial only if $\,\dim\ker\Phi^*+r^2<\dim\H_B$, i.e. if
$\,m=\dim\ker\Phi^*+r^2<+\infty$.

In this case we may assume that the first $\,n=\mathrm{rank}B\,$ vectors
of the above family $\{|\psi_j\rangle\}$ are linearly independent.
It follows that the operators $|\psi_j\rangle\langle\psi_j|$,
$j=\overline{1,n}$, generates a $\,n$-dimensional subspace $\B_n$ of
$\B(\H_A)$. Since $B\geq\sum_{j=1}^{n}|\psi_j\rangle\langle\psi_j|$
and the operator $\Phi^*(B)$ is supported by a $\,r$-dimensional
subspace $\H_r$ of $\H_A$, the operators
$\Phi^*(|\psi_j\rangle\langle\psi_j|)$ lie in $\B(\H_r)$ for
$j=\overline{1,n}$. Thus $\Phi^*(\B_n)\subseteq\B(\H_r)$ and hence
$$
\mathrm{rank}B=n=\dim\B_n\leq\dim\ker\Phi^*+\dim\B(\H_r)=\dim\ker\Phi^*+r^2=m.
$$
Since $B$ is a positive operator of rank $\leq m<+\infty$, the
finite-dimensional spectral theorem implies
$B=\sum_{j=1}^{m}|\psi'_j\rangle\langle\psi'_j|$, where
$\{|\psi'_j\rangle\}$ are orthogonal set of eigenvectors of $B$.
$\square$
\smallskip

\begin{lemma}\label{new-kraus-rep}
\emph{Let $\,\Phi(\rho)=\sum_{k=1}^{d}V_k\rho V_k^*$ be a quantum
channel and $\,\{|k\rangle\}_{k=1}^d$ an orthonormal basis in
$\,d$-dimensional Hilbert space $\,\H_d$, $d\leq+\infty$. An arbitrary
overcomplete system $\{|\psi_i\rangle\}$ of vectors in $\H_d$
generates the Kraus representation $\,\Phi(\rho)=\sum_{i}W_i\rho
W_i^*$ of the channel $\,\Phi$, where
$W_i=\sum^d_{k=1}\langle\psi_i|k\rangle V_k$.}
\end{lemma}\smallskip

\textbf{Proof.} Since $\sum_{i}|\psi_i\rangle\langle
\psi_i|=I_{\H_d}$, we have
$$
\begin{array}{c}
\displaystyle\sum_{i}W_i\rho W^*_i=\sum^d_{k,l=1} V_k\rho V^*_l
\sum_{i}\langle\psi_i|k\rangle\langle
l|\psi_i\rangle\\\displaystyle=\sum^d_{k,l=1} V_k\rho V^*_l
\sum_{i}\Tr |k\rangle\langle
l||\psi_i\rangle\langle\psi_i|=\sum^d_{k=1} V_k\rho V^*_k. \;\square
\end{array}
$$

\subsection{Orthogonal families of pure states}

The second part of Theorem \ref{main} implies the following
criterion of reversibility of a channel with respect to a given
complete family of orthogonal pure states.\smallskip

\begin{property}\label{orth-f}
\emph{Let $\,\Phi:\T(\H_A)\rightarrow\T(\H_B)$ be a quantum channel,
$\,m=\min\{m(\Phi)+1, \dim\H^{\Phi}_B\}$ \footnote{The parameter
$m(\Phi)$ and the subspace $\H^{\Phi}_B$ are defined before Theorem
\ref{main}.} and $\,\S=\{|\varphi_i\rangle\langle\varphi_i|\}$ a
complete family of orthogonal pure states in $\S(\H_A)$. The following
statements are equivalent:}
\begin{enumerate}[(i)]
  \item \emph{the channel $\,\Phi$ is reversible with  respect to the family $\,\S$;}
  \item \emph{$\,\widehat{\Phi}$ is a c-q channel having
  the representation $\,\widehat{\Phi}(\rho)=\displaystyle\sum_{i=1}^{\dim\H_A}\langle\varphi_i|\rho|\varphi_i\rangle\sigma_i$,
   where $\{\sigma_i\}$ is a set of states in $\S(\H_E)$ such that $\,\mathrm{rank}\,\sigma_i\leq m$ for all $\,i$;}

\item \emph{the channel $\,\Phi$ is isometrically equivalent to the
channel
\begin{equation*}
\Phi'(\rho)=\sum_{i,\,j=1}^{\dim\H_A}\langle\varphi_i|\rho|\varphi_j\rangle|\varphi_i\rangle\langle
\varphi_j|
\otimes\sum_{k,\,l=1}^{m}\langle\psi_{jl}|\psi_{ik}\rangle|k\rangle\langle
l|\end{equation*} from $\T(\H_A)$ into $\T(\H_A\otimes\H_m)$, where
$\{|\psi_{ik}\rangle\}$ is a collection of vectors in a separable
Hilbert space such that $\,\sum_{k=1}^{m}\|\psi_{ik}\|^2=1$ and
$\,\langle\psi_{il}|\psi_{ik}\rangle=0$ for all  $\,k\neq l$ for each
$\,i$ and $\{|k\rangle\}_{k=1}^m$ is an orthonormal basis in
$\H_m$.\footnote{Here and in what follows $\H_m$ is either $\,m$-dimensional (if $m<+\infty$) or separable (if $m=+\infty$) Hilbert space.}}
\end{enumerate}
\end{property}

\textbf{Proof.} $\mathrm{(i)\Rightarrow(ii)}$ follows from the
second part of Theorem \ref{main}, since in this case
$|\phi_i\rangle=|\varphi_i\rangle$ for all $i$.
\smallskip

$\mathrm{(ii)\Rightarrow(iii)}$.  If
$\;\sigma_i=\sum_{k=1}^{m}|\psi_{ik}\rangle\langle \psi_{ik}|\;$
then $\widehat{\Phi}(\rho)=\sum_{i,k}W_{ik}\rho\, W^*_{ik}$, where
$W_{ik}= |\psi_{ik}\rangle\langle\varphi_i|$,  and hence
representation (\ref{Kraus-rep-c}) implies
$\widehat{\widehat{\Phi}}=\Phi'$.
\smallskip

$\mathrm{(iii)\Rightarrow(i)}$ follows from Lemma \ref{isom-eq-l},
since $\,\Psi(\cdot)=\Tr_{\H_m}(\cdot)$ is a reversing channel for
the channel $\,\Phi'$ with respect to the family $\S$. $\square$
\medskip

Proposition \ref{orth-f} implies the following criterion for
reversibility of a channel in terms of its dual channel.
\smallskip

\begin{corollary}\label{orth-f-c}
\emph{A channel $\,\Phi:\T(\H_A)\rightarrow\T(\H_B)$ is reversible
with  respect to a complete family
$\{|\varphi_i\rangle\langle\varphi_i|\}$ of orthogonal pure states
in $\,\S(\H_A)$ if and only if there exists a partial isometry
$\,W:\H_A\otimes\H_m\rightarrow\H_B $ such that
\begin{equation}\label{w-cond}
|\varphi_i\rangle\langle\varphi_i|=\Phi^*(W[|\varphi_i\rangle\langle\varphi_i|\otimes
I_{\H_m}] W^*)\quad\forall i,
\end{equation}
where $\,m=\min\{m(\Phi)+1, \dim\H^{\Phi}_B\}$ and $\,\Phi^*:\B(\H_B)\rightarrow\B(\H_A)$ is the dual channel to the channel $\,\Phi$.}
\medskip
\end{corollary}

Note that  condition (\ref{w-cond}) implies $\Phi^*(WW^*)=I_{\H_A}$
and hence $WW^*$ is the projector on the subspace containing
supports of all states $\Phi(\rho)$, $\rho\in\S(\H_A)$.

\smallskip

\textbf{Proof.} Necessity of condition (\ref{w-cond}) directly
follows from Proposition \ref{orth-f}.

To prove its sufficiency consider the channel
$\Phi'(\rho)=W^*\Phi(\rho)W$ from $\T(\H_A)$ into
$\T(\H_A\otimes\H_m)$. By the remark after Corollary \ref{orth-f-c}
$$
W\Phi'(\rho)W^*=WW^*\Phi(\rho)WW^*=\Phi(\rho),\quad \rho\in\T(\H_A),
$$
and hence the channels $\Phi$ and $\Phi'$ are isometrically
equivalent. By Lemma \ref{isom-eq-l} it suffices to show reversibility
of the channel $\Phi'$ with respect to the family
$\{|\varphi_i\rangle\langle\varphi_i|\}$.\smallskip

Condition (\ref{w-cond}) implies
$$
\Tr\,[|\varphi_i\rangle\langle\varphi_i|\otimes
I_{\H_m}]\,\Phi'(|\varphi_j\rangle\langle\varphi_j|)=\Tr\,
\Phi^*(W[|\varphi_i\rangle\langle\varphi_i|\otimes I_{\H_m}]
W^*)\,|\varphi_j\rangle\langle\varphi_j|=\delta_{ij}.
$$
It follows that the support of the state
$\Phi'(|\varphi_i\rangle\langle\varphi_i|)$ belongs to the subspace
$\{\lambda|\varphi_i\rangle\}\otimes\H_m$ and hence
$\Tr_{\H_m}\Phi'(|\varphi_i\rangle\langle\varphi_i|)=|\varphi_i\rangle\langle\varphi_i|$
for all $i$. $\square$

\subsection{Arbitrary families of pure states}

In this section we consider a structure of a quantum channel reversible with respect
to an arbitrary complete family $\S=\{|\varphi_{\lambda}\rangle\langle\varphi_{\lambda}|\}_{\lambda\in\Lambda}$ of pure states.\smallskip

It is known that a channel $\Phi:\T(\H_A)\rightarrow\T(\H_B)$ is
reversible with respect to the family of all pure states in
$\S(\H_A)$ (which means that it is reversible with respect to
$\S(\H_A)$) if and only if its complementary channel is completely
depolarizing, i.e. if and only if $\Phi$ is isometrically equivalent
to the channel
\begin{equation}\label{pic}
    \Phi'(\rho)=\rho\otimes\sigma
\end{equation}
from $\T(\H_A)$ into $\T(\H_A\otimes\K)$, where $\K$ is a Hilbert
space and  $\sigma$ is a given state in $\S(\K)$ \cite[Ch.10]{H-SCI}.

We give first a characterization of a family
$\S=\{|\varphi_{\lambda}\rangle\langle\varphi_{\lambda}|\}_{\lambda\in\Lambda}\subset\S(\H_A)$
such that the reversibility of a channel
$\Phi:\T(\H_A)\rightarrow\T(\H_B)$ with respect to $\S$ implies its
reversibility with respect to $\S(\H_A)$.
\smallskip
\begin{definition}\label{orth-dec}
A family $\{|\varphi_{\lambda}\rangle\}_{\lambda\in\Lambda}$ of
vectors in $\H$ (corresp. a family
$\{|\varphi_{\lambda}\rangle\langle\varphi_{\lambda}|\}_{\lambda\in\Lambda}$
of pure states in $\S(\H)$) is called \emph{orthogonally
decomposable} if these is a proper subspace $\H_0\subset\H$ such
that some vectors of this family lie in $\H_0$ and all the others
 -- in $\H^{\perp}_0$.
\end{definition}
\smallskip
Families of pure states, which are not orthogonally decomposable, will be called \emph{orthogonally non-decomposable} (briefly, OND) families.

\medskip
\begin{property}\label{non-orth-1}
\emph{Let
$\,\{|\varphi_{\lambda}\rangle\langle\varphi_{\lambda}|\}_{\lambda\in\Lambda}$ be a complete family of pure states in $\,\S(\H_A)$. The following statements are equivalent:}
\begin{enumerate}[(i)]
  \item \emph{the family $\{|\varphi_{\lambda}\rangle\langle\varphi_{\lambda}|\}_{\lambda\in\Lambda}$
is orthogonally non-decomposable;}
 \item \emph{any channel $\,\Phi:\T(\H_A)\rightarrow\T(\H_B)$ reversible with respect to the family $\{|\varphi_{\lambda}\rangle\langle\varphi_{\lambda}|\}_{\lambda\in\Lambda}$  is isometrically equivalent to channel (\ref{pic}).}
\end{enumerate}
\end{property}

\textbf{Proof.} $\mathrm{(i)\Rightarrow(ii)}\,$ If
$\Psi:\T(\H_B)\rightarrow\T(\H_A)$ is a reversing channel for the
channel $\Phi$ then Lemma \ref{t-c} below shows that
$\Psi\circ\Phi=\id_{\H_A}$. Thus the channel $\Phi$ is reversible
with respect to the set $\S(\H_A)$ and hence its complementary
channel $\widehat{\Phi}$ is completely depolarizing.\smallskip

$\mathrm{(ii)\Rightarrow(i)}\,$ If $\H_0$ is a proper
subspace of $\H_A$ such that  the vector $|\varphi_{\lambda}\rangle$
lies either in $\H_0$ or in $\H^{\perp}_0$ for each $\lambda\in\Lambda$
then the channel $\rho\mapsto P_0\rho P_0+\bar{P}_0\rho \bar{P}_0$, where $P_0$ is the projector on $\H_0$ and $\bar{P}_0=I_{\H_A}-P_0$, is obviously reversible with respect to the family $\{|\varphi_{\lambda}\rangle\langle\varphi_{\lambda}|\}_{\lambda\in\Lambda}$. $\square$
\smallskip
\begin{lemma}\label{t-c}
\emph{Let $\,\Phi:\T(\H)\rightarrow\T(\H)$ be a quantum channel
($\dim\H\leq+\infty$) and
$\{|\varphi_{\lambda}\rangle\langle\varphi_{\lambda}|\}_{\lambda\in\Lambda}$
be an orthogonally non-decomposable family of pure states in
$\,\S(\H)$. If
$\;\Phi(|\varphi_{\lambda}\rangle\langle\varphi_{\lambda}|)=|\varphi_{\lambda}\rangle\langle\varphi_{\lambda}|\,$
for all $\lambda\in\Lambda$ then $\,\Phi|_{\T(\H_0)}=\id_{\H_0}$,
where $\H_0$ is the subspace generated by the family
$\{|\varphi_{\lambda}\rangle\}_{\lambda\in\Lambda}$.}
\end{lemma}
\smallskip
\textbf{Proof.} Let $\Phi(\rho)=\Tr_{\K}V\rho V^*$ be the
Stinespring representation of the channel $\Phi$, where $V$ is an
isometry from $\H$ into $\H\otimes\K$.

By using the standard argumentation based on Zorn's lemma one can show that any complete OND  family of pure states contains a countable complete OND subfamily (Lemma \ref{app} in Appendix 5.2).

Let $\{|\varphi_{i}\rangle\langle\varphi_{i}|\}$ be a countable OND subfamily of
$\{|\varphi_{\lambda}\rangle\langle\varphi_{\lambda}|\}_{\lambda\in\Lambda}$ such that the family $\{|\varphi_{i}\rangle\}$ generates the subspace $\H_0$. The condition of the lemma
implies
$$
V|\varphi_i\rangle=|\varphi_i\rangle\otimes|\psi_i\rangle,\quad
\forall i,
$$
where $\{|\psi_i\rangle\}$ is a family of unit vectors in $\K$.
Since $V$ is an isometry, we have
$$
\langle\varphi_i|\varphi_j\rangle=\langle V\varphi_i|V\varphi_j\rangle=\langle\varphi_i|\varphi_j\rangle\langle\psi_i|\psi_j\rangle,\quad
\forall i,j
$$
and hence
$\langle\varphi_i|\varphi_j\rangle\neq0\;\Rightarrow\;\langle\psi_i|\psi_j\rangle=1$.

It follows that  $|\psi_i\rangle=|\psi_j\rangle$ for all $i,j$.
Indeed, if there exist index sets $I$ and $J$ such that
$|\psi_i\rangle\neq|\psi_j\rangle$ for all $i\in I,j\in J$ then the
above implication shows that $\langle\varphi_i|\varphi_j\rangle=0$
for all $i\in I,j\in J$ contradicting to the assumed
orthogonal non-decomposability of the family
$\{|\varphi_i\rangle\langle\varphi_i|\}$.

Thus we have
$V|\varphi_i\rangle=|\varphi_i\rangle\otimes|\psi\rangle$ for all
$i$ and hence $V|\varphi\rangle=|\varphi\rangle\otimes|\psi\rangle$
for all $|\varphi\rangle\in\H_0$, since the family
$\{|\varphi_i\rangle\}$ generates the subspace $\H_0$. It follows
that $\Phi(\rho)=\rho$ for all $\rho\in\T(\H_0)$. $\square$
\medskip

In analysis of reversibility of a channel with respect to orthogonally decomposable families of pure states the following simple observation plays an essential role.
\smallskip
\begin{lemma}\label{nod-dec}
\emph{An arbitrary family $\,\S$ of pure states in $\,\S(\H)$ can be
decomposed as follows $\,\S=\bigcup_k\S_k$, where $\{\S_k\}$ is a
finite or countable collection of OND disjoint subfamilies of $\,\S$
such that $\rho\perp\rho'$ for all $\rho\in\S_k,\rho'\in\S_{k'}$,
$k\neq k'$. This decomposition is unique (up to permutation of the
subfamilies).}
\end{lemma}
\smallskip

\textbf{Proof.} For given $\rho\in\S$ consider the monotone sequence $\{\C^{\rho}_n\}$ of
subfamilies of $\S$ constructed as follows. Let
$\C^{\rho}_1=\{\rho\}$, $\C_2$ be the family of all
states from $\S$ non-orthogonal to $\rho$,
$\C_{n+1}$ be the family of all states from $\S$ non-orthogonal to at
least one state from $\C_n$, $n=2,3,...$. Let $\C^{\rho}_*=\bigcup_n\C^{\rho}_n$. It is easy to
verify by induction that $\C^{\rho}_n$ is an OND family for each $n$ and hence
$\C^{\rho}_*$ is an OND family. Note that any state in $\C^{\rho}_*$ is orthogonal to
any state in $\S\setminus\C^{\rho}_*$. Indeed, if $\rho\in\C^{\rho}_*$
then $\rho\in\C^{\rho}_n$ for some $n$. So, if a pure state $\sigma$ is
not orthogonal to $\rho$ then it belongs to $\C^{\rho}_{n+1}\subseteq\C^{\rho}_*$.

It is easy to see that the families $\C^{\rho}_*$ and
$\C^{\rho'}_*$, $\rho,\rho'\in\S$,  either coincide or have an empty intersection.
Since the Hilbert space $\H$ is separable and each family
$\C^{\rho}_*$ occupies a nontrivial subspace of $\H$, the collection $\{\C^{\rho}_*\}_{\rho\in\S}$ contains either a finite or countable number of different families. These families form the required decomposition. $\square$
\medskip

The above decomposition of a complete family $\S$ of pure states provides a
description of the class of all channels reversible with respect to
$\S$.
\smallskip

\begin{theorem}\label{non-orth-2}
\emph{Let $\,\Phi:\T(\H_A)\rightarrow\T(\H_B)$ be a quantum channel,
$\,m=\min\{m(\Phi)+1, \dim\H^{\Phi}_B\}$ \footnote{The parameter
$m(\Phi)$ and the subspace $\H^{\Phi}_B$ are defined before Theorem
\ref{main}.} and $\,\S$ a complete family of pure states in
$\S(\H_A)$. Let $\,\S=\bigcup_{k=1}^n\S_k$, $n\leq\dim\H_A$, be a
decomposition of $\,\S$ into OND subfamilies (from Lemma
\ref{nod-dec}) and $\,P_k$ -- the projector on the subspace
generated by the states in $\,\S_k$. The following statements are
equivalent:}
\begin{enumerate}[(i)]
  \item \emph{the channel $\,\Phi$ is reversible with respect to the family $\,\S$;}
  \item \emph{the channel $\,\Phi$ is reversible with respect to the family}
  $$
  \hat{\S}=\left\{\rho\in\S(\H_A)\,\left|\;\rho=\sum_{k=1}^n P_k\rho P_k\right.\right\};
  $$\vspace{-15pt}
  \item \emph{$\,\widehat{\Phi}$ is a c-q channel having the representation
  $\,\widehat{\Phi}(\rho)=\displaystyle\sum_{k=1}^n[\Tr P_k\rho]\sigma_k$, where $\{\sigma_k\}$ is
  a set of states in $\S(\H_E)$ such that $\,\mathrm{rank}\sigma_k\!\leq m$ for all $\,k;$}
  \item \emph{the channel $\,\Phi$ is isometrically equivalent to the
  channel}
$$
\Phi'(\rho)=\sum_{k,l=1}^nP_k\rho
P_l\otimes\sum_{p,t=1}^m\langle\psi_t^l|\psi_p^k\rangle
|p\rangle\langle t|\vspace{-5pt}
$$
\emph{from $\,\T(\H_A)$ into $\,\T(\H_A\otimes\H_m)$, where
$\{|\psi_{p}^k\rangle\}$ is a collection of vectors  in a separable
Hilbert space such that $\,\sum_{p=1}^m\|\psi_{p}^k\|^2=1$ and
$\langle\psi_{t}^k|\psi_{p}^k\rangle=0$ for all $\,p\neq t$ for each
 $k$ and $\,\{|p\rangle\}_{p=1}^m$ is an orthonormal basis in $\,\H_m$.}
\end{enumerate}
\end{theorem}

\textbf{Proof.}  $\mathrm{(i)\Rightarrow(ii)}$. Let $\Psi$ be a
channel such that $\Psi(\Phi(\rho))=\rho$ for all $\rho\in\S$. Let
$\H_k$ be the subspace of $\H$ generated by the vectors
corresponding to the subfamily $\S_k$. Since $\S_k$ is an OND
family, Lemma \ref{t-c} shows that
$\Psi\circ\Phi|_{\T(\H_k)}=\id_{\H_k}$ for each $k$. \smallskip

$\mathrm{(ii)\Rightarrow(iii)}$. Let $\{|\phi_i\rangle\}$ be an
orthonormal basis corresponding to the decomposition
$\H_A=\oplus_k\H_k$, i.e. each $|\phi_i\rangle$ lies in some $\H_k$.
Let $I_k$ be the set of all $i$ such that $|\phi_i\rangle\in\H_k$.
Since $|\phi_i\rangle\langle\phi_i|\in\hat{\S}$ for all $i$, the
channel $\Phi$ is reversible with respect to the family
$\{|\phi_i\rangle\langle\phi_i|\}$. By Proposition \ref{orth-f} we
have
$$
\widehat{\Phi}(\rho)=\sum_{k}\sum_{i\in
I_k}\langle\phi_i|\rho|\phi_i\rangle \sigma_i,
$$
where $\{\sigma_i\}$ is a set of states in $\,\S(\H_E)$ such that
$\,\mathrm{rank}\,\sigma_i\leq m$ for all $i$. Since $\S_k$ is an OND
family, Proposition \ref{non-orth-1} shows that the restriction of
the channel $\widehat{\Phi}$ to the set $\T(\H_k)$ is a completely
depolarizing channel. Hence $\sigma_i=\bar{\sigma}_k$ for all $i\in
I_k$. Thus
$\widehat{\Phi}(\rho)=\sum_{k}[\Tr P_k\rho]\bar{\sigma}_k$. \medskip

$\mathrm{(iii)\Rightarrow(iv)}$. Let $k(i)$ be the index of the set
$I_k$ containing $i$, i.e. $i\in I_{k(i)}$ for all $i$. If
$\;\sigma_k=\sum_{p=1}^{m}|\psi^k_{p}\rangle\langle \psi^k_{p}|\;$
then $\widehat{\Phi}(\rho)=\sum_{i,p}W_{ip}\rho\, W^*_{ip}$, where
$W_{ip}= |\psi^{k(i)}_{p}\rangle\langle\phi_i|$, and hence
representation (\ref{Kraus-rep-c}) implies
$$
\begin{array}{c}
\displaystyle\widehat{\widehat{\Phi}}(\rho)=\sum_{i,j,p,t}[\Tr
W_{ip}\rho\,
W^*_{jt}]|\phi_i\rangle\langle \phi_j|\otimes|p\rangle\langle t|=\\
\displaystyle\sum_{k,l,p,t}\,\sum_{i\in I_k,j\in I_l}
\langle\phi_i|\rho|\phi_j\rangle|\phi_i\rangle\langle
\phi_j|\otimes \langle\psi_t^l|\psi_p^k\rangle |p\rangle\langle
t|=\sum_{k,l}P_k\rho
P_l\otimes\sum_{p,t}\langle\psi_t^l|\psi_p^k\rangle
|p\rangle\langle t|,
\end{array}
$$
where $\,\{|p\rangle\}$ is an orthonormal basis in $\H_m$.

$\mathrm{(iv)\Rightarrow(i)}$ follows from Lemma \ref{isom-eq-l},
since $\,\Psi(\cdot)=\Tr_{\H_m}(\cdot)$ is a reversing channel for
the channel $\,\Phi'$ with respect to the family $\S$. $\square$

\medskip
Theorem \ref{non-orth-2} implies the following useful observation.
\smallskip

\begin{corollary}\label{non-orth-2-c}
\emph{If a channel $\,\Phi:\T(\H_A)\rightarrow\T(\H_B)$  is
reversible with respect to a complete family $\,\S$ of pure states
in $\,\S(\H_A)$ then it is reversible with respect to a particular
complete family of orthogonal pure states in $\,\S(\H_A)$ and hence $\,\dim\H_A\leq\dim\H_B$.}
\end{corollary}
\smallskip

\begin{remark}\label{non-orth-2-r}
If the  complete family of pure states $\S$ contains a subfamily
$\S_0=\{|\varphi_i\rangle\langle \varphi_i|\}$ such that
$\{|\varphi_i\rangle\}$ is a basis in the space $\H_A$ (in the sense
that an arbitrary vector $|\psi\rangle$ has a unique decomposition
$|\psi\rangle=\sum_i c_i|\varphi_i\rangle$)\footnote{Existence of
the subfamily $\S_0$ is obvious if $\H_A$ is a finite-dimensional space.
The condition showing that a complete countable
family of unit vectors in an infinite-dimensional Hilbert space forms
a basis can be found in \cite[Chapter I]{A&G}.} then the family of
orthogonal pure states mentioned in Corollary \ref{non-orth-2-c} is
explicitly given by Theorem \ref{main}. Indeed, by Lemma \ref{onb}
in Appendix 5.2 the set $\{|\phi_i\rangle\}$ of vectors defined in
(\ref{v-phi-rep}) by means of an arbitrary non-degenerate
probability distribution $\{\pi_i\}$ forms an orthonormal basis in
$\H_A$. By Proposition \ref{orth-f} the channel $\Phi$ is reversible
with respect to the family $\{|\phi_i\rangle\langle \phi_i|\}$.
\end{remark}
\medskip

There are two cases in which the reversibility criterion from Theorem \ref{non-orth-2} is simplified to the limit.

\smallskip

\begin{corollary}\label{non-orth-2-c+}
\emph{Let $\,\Phi:\T(\H_A)\rightarrow\T(\H_B)$ be a quantum channel
satisfying one of the following conditions:
\begin{itemize}
    \item $\ker\Phi^*=\{0\}$.
    \item $\dim\H_A=\dim\H_B<+\infty$,
\end{itemize}
Let $\,\S$ be a complete family of pure states in $\,\S(\H_A)$,
$\,\S=\bigcup_{k=1}^n\S_k$ its decomposition into OND subfamilies
($n\leq\dim\H_A$) and $P_k$ the projector on the subspace generated
by the states in $\,\S_k$. The channel $\,\Phi$ is reversible with
respect to the family $\,\S$ if and only if it is unitary equivalent
to the channel
$$
\Phi'(\rho)=\sum_{k,l=1}^n c_{kl}P_k\rho P_l
$$
from $\T(\H_A)$ into itself, where $\|c_{kl}\|$ is a Gram matrix of
a collection of unit vectors (in the case $\,\ker\Phi^*=\{0\}$ this matrix contains no zeros).}
\end{corollary}
\smallskip
\textbf{Proof.} If  $\ker\Phi^*=\{0\}$ then $m(\Phi)=1$ and the
assertion of the corollary directly follows from Theorem
\ref{non-orth-2}. We have only to note that in this case
$\H_B^{\Phi}=\H_B^{\Phi'}=\H_B$ and hence isometrical equivalence of
the channels $\Phi$ and $\Phi'$ means their unitary equivalence.

Consider the case $d=\dim\H_A=\dim\H_B<+\infty$. By Corollary
\ref{non-orth-2-c} the reversibility of the channel $\Phi$ with
respect to $\S$ implies its reversibility with respect to some
family $\{\rho_i\}_{i=1}^d$ of orthogonal pure states in $\S(\H_A)$.
Hence
$$
\frac{1}{d}\sum_{i=1}^d
H(\Phi(\rho_i)\|\Phi(\bar{\rho}))=\frac{1}{d}\sum_{i=1}^d
H(\rho_i\|\bar{\rho})=\log d,
$$
where $\bar{\rho}=d^{-1}I_{\H_A}$. It follows that the family
$\{\Phi(\rho_i)\}_{i=1}^d$ consists of orthogonal pure states  and
that $\Phi(I_{\H_A})=I_{\H_B}$.

Hence, by definition of the complementary channel,
$\{\widehat{\Phi}(\rho_i)\}_{i=1}^d$ is a family of pure states and
Theorem \ref{non-orth-2} shows that
$\,\widehat{\Phi}(\rho)=\displaystyle\sum_{k}[\Tr
P_k\rho]|\psi_k\rangle\langle\psi_k|$, where $\{|\psi_k\rangle\}$ is
a set of unit vectors in $\H_E$. If follows that the channel $\Phi$
is isometrically equivalent to the channel
$\widehat{\widehat{\Phi}}=\Phi'$ with
$c_{kl}=\langle\psi_l|\psi_k\rangle$. Since the both channels are
unital, their isometrical equivalence means unitary equivalence.
$\square$ \smallskip

\begin{remark}\label{non-orth-2-c+r}
If one of the conditions of Corollary \ref{non-orth-2-c+r} holds for
a channel $\Phi$ then this channel is reversible with respect to a
complete family $\S$ of pure states if and only if $\Phi(\rho)=U\rho
U^*$ for all $\rho\in\S$, where $U$ is an unitary operator, i.e.
reversibility of the channel $\Phi$ with respect to a complete
family of pure states is equivalent to \emph{preserving} of all
states of the family by this channel (up to unitary transformation).
\end{remark}

\section{Conditions  for preserving the Holevo quantity and their applications}\label{s5}

Consider some applications of the results of Section 3 in quantum
information theory.\smallskip

A finite or countable collection of states $\{\rho_i\}$ with the
corresponding probability distribution $\{\pi_i\}$ is called
\emph{ensemble} and denoted $\{\pi_i,\rho_i\}$. The state
$\bar{\rho}=\sum_i \pi_i\rho_i$ is called the \emph{average state}
of the ensemble $\{\pi_i,\rho_i\}$.\smallskip

The Holevo quantity of an ensemble $\{\pi_i,\rho_i\}$ is defined as
follows
\begin{equation*}
\chi(\{\pi_i,\rho_i\})\doteq\sum_i\pi_i
H(\rho_i\|\bar{\rho})=H(\bar{\rho})-\sum_i\pi_i H(\rho_i),
\end{equation*}
where the second formula is valid under the condition
$H(\bar{\rho})<+\infty$. This value plays a central role in analysis
of different protocols of classical information transmissions by a
quantum channel \cite{H-SCI,N&Ch}.

By monotonicity of the quantum relative entropy we have
\begin{equation}\label{chi-q-m}
\chi(\{\pi_i,\Phi(\rho_i)\})\leq\chi(\{\pi_i,\rho_i\}).
\end{equation}
for an arbitrary quantum channel $\Phi:\T(\H_A)\rightarrow\T(\H_B)$
and any ensemble $\{\pi_i,\rho_i\}$ of states in $\S(\H_A)$.

\begin{remark}\label{ent-g-c}
If $H(\bar{\rho})<+\infty$ and  $H(\Phi(\bar{\rho}))<+\infty$ then
inequality (\ref{chi-q-m}) means convexity of the function
$\rho\mapsto H(\Phi(\rho))-H(\rho)$ -- the entropy gain of the channel
$\Phi$.
\end{remark}\smallskip

By Theorem \ref{P-th} an
equality in (\ref{chi-q-m}) under the condition
$\chi(\{\pi_i,\rho_i\})<+\infty$ is equivalent to reversibility of
the channel $\Phi$ with respect to the family $\{\rho_i\}$. Thus, the results of Section 3 provide conditions of this
equality (which can be interpreted as preserving the Holevo
quantity of the ensemble $\{\pi_i,\rho_i\}$ under the channel $\Phi$).
\smallskip

In  analysis of infinite-dimensional quantum systems and channels it
is necessary to consider \emph{generalized} (or \emph{continuous})
ensembles, defined as  Borel probability measures on the set of
quantum states (from this point of view ensemble
$\{\pi_i,\rho_i\}$ is the purely atomic  measure
$\sum_i\pi_i\delta_{\rho_i}$, where $\delta_{\rho}$ is a Dirac
measure concentrated at a state $\rho$) \cite{H-SCI,H-Sh-2}.
\smallskip

The Holevo quantity of a generalized ensemble (measure) $\mu$ is defined as
follows
\begin{equation}\label{h-q-c}
\chi(\mu)=\int_{\mathfrak{S}(\mathcal{H})}H(\rho\Vert\bar{\rho}(\mu))\mu(d\rho),
\end{equation}
where $\bar{\rho}(\mu)$ is the barycenter of $\,\mu$
defined by the Bochner integral
$$
\bar{\rho}(\mu)=\int_{\mathfrak{S}(\mathcal{H})}\rho \mu(d\rho ).
$$
If $\,H(\bar{\rho}(\mu))<+\infty\,$ then
$\,\chi(\mu)=H(\bar{\rho}(\mu))-\int_{\mathfrak{S}(\mathcal{H})}H(\rho)\mu
(d\rho)\,$ \cite{H-Sh-2}.\smallskip

Denote by $\P(\mathcal{A})$ the set of all Borel probability
measures on a closed subset $\mathcal{A}\subseteq\S(\H)$ endowed with
the weak convergence topology \cite{Par}.\smallskip

The image of a generalized ensemble $\mu\in\P(\S(\H_A))$ under a
channel $\Phi:\T(\H_A)\rightarrow\T(\H_B)$ is a generalized ensemble
$\Phi(\mu)\doteq\mu\circ\Phi^{-1}\in\P(\S(\H_B))$. Its Holevo
quantity can be expressed as follows
\begin{equation}\label{chi-phi-mu}
\begin{array}{c}
\displaystyle\chi(\Phi(\mu))=\int_{\mathfrak{S}(\H_A)}H(\Phi(\rho
)\Vert \Phi (\bar{\rho}(\mu)))\mu(d\rho )\\\displaystyle=H(\Phi
(\bar{\rho}(\mu)))-\int_{\mathfrak{S}(\H_A)}H(\Phi(\rho))\mu
(d\rho),
\end{array}
\end{equation}
where the second formula is valid under the condition $H(\Phi
(\bar{\rho}(\mu)))<+\infty$.\smallskip

Similarly to the discrete case monotonicity of the relative entropy
implies monotonicity of the Holevo quantity for generalized
ensembles:
\begin{equation}\label{chi-d++}
\chi(\Phi(\mu))\leq\chi(\mu).
\end{equation}

Theorem \ref{P-th} implies the following criterion of an equality in
(\ref{chi-d++}), which is a modification of Theorem 3 in \cite{J&P}
(in the case $\mathcal{M}=\B(\H)$).
\smallskip
\begin{property}\label{chi-eq}
\emph{Let $\,\Phi:\T(\H_A)\rightarrow\T(\H_B)$ be a quantum channel
and $\mu$  a generalized ensemble in $\P(\S(\H_A))$ such that
$\,\chi(\mu)<+\infty$. The following statements are equivalent:}
\begin{enumerate}[(i)]
  \item $\chi(\Phi(\mu))=\chi(\mu)$;
  \item \emph{$H(\Phi(\rho)\Vert \Phi (\bar{\rho}(\mu)))=H(\rho
\Vert \bar{\rho}(\mu))$ for $\mu$-almost all $\rho$ in
$\,\S(\H_A)$;}
  \item \emph{$\rho=\Theta_{\bar{\rho}(\mu)}(\Phi(\rho))$ for $\mu$-almost
all $\rho$ in $\,\S(\H_A)$;}
  \item \emph{the channel $\,\Phi$ is reversible with respect to $\mu$-almost
all $\rho$ in $\,\S(\H_A)$.}
\end{enumerate}
\end{property}

In contrast to Theorem 3 in \cite{J&P}, in Proposition \ref{chi-eq}
it is not assumed that the "dominating" state $\bar{\rho}(\mu)$ is a
countable convex mixture of some states from the support of $\mu$.
\medskip

By Proposition \ref{chi-eq} Theorem \ref{main} (with Lemma 2 in \cite{J&P}) and Theorem \ref{non-orth-2} imply the following
conditions for equality in (\ref{chi-d++}).
\smallskip

\begin{theorem}\label{main++} \emph{Let $\,\Phi:\T(\H_A)\rightarrow\T(\H_B)$ be a quantum
channel.} \emph{If there exists an ensemble $\mu\in\P(\S^r)$, where
$\,\S^r=\{\rho\in\S(\H_A)\,|\,\mathrm{rank}\,\rho\leq r\}$, with
full rank average state $\,\bar{\rho}(\mu)$ such that
\begin{equation}\label{chi-nd++}
\chi(\Phi(\mu))=\chi(\mu)<+\infty,
\end{equation}
then the complementary channel $\,\widehat{\Phi}$ has Kraus
representation (\ref{Kraus-rep}) consisting of $\,\leq
n\times\min\{m(\Phi)+r^2, \dim\H^{\Phi}_B\}$ summands \footnote{The parameter
$m(\Phi)$ and the subspace $\H^{\Phi}_B$ are defined before Theorem
\ref{main}.} such that
$\;\mathrm{rank}V_k\leq r$ for all $\,k$ and hence
$\,\widehat{\Phi}$ is a $\,r$-partially entanglement-breaking channel
(Def.\ref{p-e-b-ch-d})}.\smallskip

\emph{If the above hypothesis holds with $\,r=1$ then equivalent
statements $\mathrm{(i)\textup{-}(iv)}$ of Theorem \ref{non-orth-2}
are valid for the channel $\,\Phi$ with an orthogonal
resolution of the identity $\,\{P_k\}$ such that
$\,\rho=\sum_{k}P_k\rho P_k\,$ for $\mu$-almost all $\,\rho$ in
$\,\S(\H_A)$.}\footnote{More precisely, $\{P_k\}$ is the
minimal orthonormal resolution of the identity possessing this property.}
\end{theorem}
\smallskip

We consider below some corollaries of this theorem related to
different characteristics of quantum systems and channels.

\subsection{The Holevo capacity and the minimal output entropy of a finite-dimensional channel}

Let $\Phi:\T(\H_A)\rightarrow\T(\H_B)$ be a channel between
finite-dimensional quantum systems
($\dim\H_A,\dim\H_B<+\infty$).\smallskip

The Holevo capacity of the channel $\Phi$ is defined as follows
(cf.\cite{H-SCI,N&Ch})
\begin{equation}\label{chi-cap}
\bar{C}(\Phi)=\sup_{\{\pi_i,\rho_i\}}\chi(\{\pi_i,\Phi(\rho_i)\}).
\end{equation}
It follows from inequality (\ref{chi-q-m}) that
\begin{equation}\label{chi-c-eq}
\bar{C}(\Phi)\leq\log\dim\H_A.
\end{equation}
Since the supremum in (\ref{chi-cap}) is always achieved at some
ensembles of pure states \cite{Sch-West}, Theorem \ref{main++} (with
$r=1$) and Corollary \ref{non-orth-2-c+} imply the following
criteria of an equality in (\ref{chi-c-eq}). \smallskip
\begin{corollary}\label{main-c-2}
A) \emph{An equality holds in (\ref{chi-c-eq}) if and only if
equivalent statements $\mathrm{(i)\textup{-}(iv)}$ of Theorem
\ref{non-orth-2} are valid for the channel $\,\Phi$ with a
particular orthogonal resolution of the identity
$\{P_k\}$.}\smallskip

B) \emph{If $\,\H_B=\H_A\,$ then an equality holds in
(\ref{chi-c-eq}) if and only if the channel $\,\Phi$ is unitary
equivalent to the channel $\,\Phi'$ described in Corollary
\ref{non-orth-2-c+} with a particular orthogonal resolution of the
identity $\{P_k\}$.}
\end{corollary}\smallskip

Corollary \ref{main-c-2}B implies the following observation concerning the
minimal output entropy
$$
H_{\mathrm{min}}(\Phi)=\min_{\rho\in\S(\H_A)}H(\Phi(\rho))
$$
of covariant channels.\medskip

\begin{corollary}\label{main-c-3}
\emph{Let $\,\Phi:\T(\H_A)\rightarrow\T(\H_B)$, $\H_B=\H_A$, be a
quantum channel covariant with respect to some irreducible
representation $\{V_g\}_{g\in G}$ of a compact group $G$ in the
sense that $\,\Phi(V_g\rho V^*_g)=V_g\Phi(\rho)V^*_g$ for all $\,g\in
G$. The equality $\,H_{\mathrm{min}}(\Phi)=0\,$ holds if and only if
the channel $\,\Phi$ is unitary equivalent to the channel $\,\Phi'$
described in Corollary \ref{non-orth-2-c+} with a
particular orthogonal resolution of the identity $\{P_k\}$.}
\end{corollary}

\begin{proof} It is sufficient to note that the covariance condition implies
$\bar{C}(\Phi)=\log\dim\H_B-H_{\mathrm{min}}(\Phi)$ \cite{H-r-c-c}.
\end{proof}

Corollary \ref{main-c-3} gives a criterion of the equality
$\,H_{\mathrm{min}}(\Phi)=0\,$ for any unital qubit channel $\Phi$
(for which $\dim\H_A=\dim\H_B=2$ and  $\Phi(I_{\H_A})=I_{\H_B}$)
\cite{H-SCI}.

\subsection{Strict decrease of the Holevo quantity under partial
trace and strict concavity of the quantum conditional entropy}

Since the partial trace
$\T(\H\otimes\K)\ni\rho\mapsto\Tr_{\H}\rho\in\T(\K)\,$ is not a
$\,r$-PEB channel for $\,r<\dim\K$, Theorem \ref{main++} implies the
following observations.
\smallskip

\begin{property}\label{main-c-4} \emph{Let $\,\H_A=\H_B\otimes\H_E$ and
$\;\Phi(\rho)=\Tr_{\H_E}\rho$, $\rho\in\S(\H_A)$.}\smallskip

A) \emph{$\chi(\{\pi_i,\Phi(\rho_i)\})<\chi(\{\pi_i,\rho_i\})$ for
any ensemble $\{\pi_i,\rho_i\}$ of states in $\,\S(\H_A)$ with
full rank average state such that
$\;\mathrm{sup}_i\,\mathrm{rank}\rho_i<\dim \H_E\,$ and
$\,\chi(\{\pi_i,\rho_i\})<+\infty$.}\smallskip

B) \emph{$\chi(\Phi(\mu))<\chi(\mu)$ for any generalized ensemble $\mu$ in
$\,\P(\S(\H_A))$ with the full rank average state $\,\bar{\rho}(\mu)$ such that
$\;\mathrm{sup}_{\rho\in\mathrm{supp}\mu}\,\mathrm{rank}\rho<\dim
\H_E\,$ and $\,\chi(\mu)<+\infty$.}
\end{property}\smallskip

\begin{remark}\label{full-rank-n} By the Stinespring representation every quantum channel is unitary equivalent to a
particular subchannel of a partial trace. Since the Holevo quantity
does not strict decrease for all channels, Proposition
\ref{main-c-4} clarifies necessity of the full rank average state
condition in Theorem \ref{main++}. $\square$
\end{remark}\medskip

The quantum conditional entropy of a state $\rho$ of a composite system $AB$
is defined as follows
$$
H_{A|B}(\rho)\doteq H(\rho)-H(\Tr_{\H_A}\rho)
$$
provided
\begin{equation}\label{f-cond}
H(\rho)<+\infty\quad \textrm{and}\quad H(\Tr_{\H_A}\rho)<+\infty.
\end{equation}

By Remark \ref{ent-g-c} concavity of the function $\rho\mapsto
H_{A|B}(\rho)$ on the convex set defined by conditions (\ref{f-cond})
follows from monotonicity of the Holevo quantity under partial trace. Proposition
\ref{main-c-4}A implies the following strict concavity property of
the conditional entropy.\medskip

\begin{corollary}\label{main-c-5} \emph{Let $\rho$ be a full rank state in $\,\S(\H_{A}\otimes\H_{B})$ satisfying (\ref{f-cond}).
Then}
$$
H_{A|B}(\rho)>\sum_i \pi_i H_{A|B}(\rho_i)
$$
\emph{for any ensemble $\{\pi_i,\rho_i\}$ with the average state
$\rho$ such that $\;\mathrm{rank}\rho_i<\dim\H_A$ for all $\,i$.}
\end{corollary}\medskip

By using Proposition \ref{main-c-4}B one can obtain a continuous
(integral) version of Corollary \ref{main-c-5}.

It is easy to construct an example showing that the  strict
concavity property of the conditional entropy stated in Corollary
\ref{main-c-5} does not hold for arbitrary state $\rho$ and its
convex decomposition. \bigskip

Theorem \ref{main++} is essentially used in the proof of the criterion of
an equality between the constrained Holevo capacity (the
$\chi$-function) and the quantum mutual information of a quantum
channel \cite{TEC}.

\section{Appendix}

\subsection{Proof of Petz's theorem (Theorem \ref{P-th}) for degenerate states}

It suffices to prove $\mathrm{(i)\Rightarrow(iii)}$ assuming that $\rho$ is an arbitrary state and $\sigma$
is a full rank state.\footnote{I would
be grateful for any reference on the proof of Theorem \ref{P-th} in
infinite dimensions without the full rank condition on the state
$\rho$.} Consider the ensemble consisting of two states $\rho$ and
$\sigma$ with probabilities $t$ and $1-t$, where $t\in(0,1)$. Let
$\sigma_{t}=t\rho+(1-t)\sigma$. By Donald's identity (Proposition
5.22 in \cite{O&P}) we have
\begin{equation}\label{d-one}
t H(\rho\|\hspace{1pt}\sigma)+(1-t)H(\sigma\|\hspace{1pt}\sigma)= t
H(\rho\|\hspace{1pt}\sigma_{t})+(1-t)H(\sigma\|\hspace{1pt}\sigma_{t})+H(\sigma_{t}\|\hspace{1pt}\sigma)
\end{equation}
and
\begin{equation}\label{d-two}
\begin{array}{c}
t
H(\Phi(\rho)\|\Phi(\sigma))+(1-t)H(\Phi(\sigma)\|\Phi(\sigma))\\\\=
t
H(\Phi(\rho)\|\Phi(\sigma_{t}))+(1-t)H(\Phi(\sigma)\|\Phi(\sigma_{t}))+H(\Phi(\sigma_{t})\|\Phi(\sigma)),
\end{array}
\end{equation}
where the left-hand sides are finite and coincide by the condition.
Since the first, the second and the third terms in the right-hand
side of (\ref{d-one}) are not less than the corresponding terms in
(\ref{d-two}) by monotonicity of the relative entropy, we obtain
\begin{equation}\label{d-three}
H(\Phi(\rho)\|\Phi(\sigma_{t}))=H(\rho\|\hspace{1pt}\sigma_{t})\quad
\textrm{and}\quad
H(\Phi(\sigma)\|\Phi(\sigma_{t}))=H(\sigma\|\hspace{1pt}\sigma_{t}).
\end{equation}
It follows from \cite[Theorem 3 and Proposition 4]{J&P} that
$\rho=\Theta_{t}(\Phi(\rho))$ for all $t\in(0,1)$, where
$$
\Theta_{t}(\varrho\,)=[\sigma_{t}]^{1/2}\Phi^*\left([\Phi(\sigma_{t})]^{-1/2}(\varrho\,)[\Phi(\sigma_{t})]^{-1/2}\right)[\sigma_{t}]^{1/2},
\quad \varrho\in\S(\H_B).
$$

To complete the proof it suffices to show that
\begin{equation}\label{s-lim-one}
\lim_{t\rightarrow +0}\Theta_{t}=\Theta_{\sigma}
\end{equation}
in the strong convergence topology (in which $\Phi_n\rightarrow\Phi$
means $\Phi_n(\rho)\rightarrow\Phi(\rho)$ for all $\rho$
\cite{Sh-H}), since this implies $\rho=\lim_{t\rightarrow
+0}\Theta_{t}(\Phi(\rho))=\Theta_{\sigma}(\Phi(\rho))$.

Since $\Theta_{t}(\Phi(\sigma))=\sigma$ for all $t\in(0,1)$, the set
of channels $\{\Theta_{t}\}_{t\in(0,1)}$ is relatively compact in
the strong convergence topology by Corollary 2 in \cite{Sh-H}. Hence
there exists a sequence $\{t_n\}$ converging to zero such that
\begin{equation}\label{s-lim-two}
\lim_{n\rightarrow+\infty}\Theta_{t_n}=\Theta_0,
\end{equation}
where $\Theta_0$ is a particular channel. We will show that
$\Theta_0=\Theta_{\sigma}$.

Note that (\ref{s-lim-two}) means that the sequence
$\{\Theta^*_{t_n}(A)\}$ tends to the operator $\Theta^*_{0}(A)$ in
the weak operator topology for any positive
$A\in\B(\H_A)$.\footnote{Since this topology coincides with the
$\sigma$-weak operator topology on the unit ball of $\B(\H_B)$
\cite{B&R}.} By Lemma \ref{l-one} below we have
$$
\lim_{n\rightarrow+\infty}[\Phi(\sigma_{t_n})]^{1/2}\Theta^*_{t_n}(A)[\Phi(\sigma_{t_n})]^{1/2}=
[\Phi(\sigma)]^{1/2}\Theta^*_{0}(A)[\Phi(\sigma)]^{1/2}
$$
in the Hilbert-Schmidt norm topology. But the explicit form of
$\Theta^*_{t_n}$ shows that
$$
[\Phi(\sigma_{t_n})]^{1/2}\Theta^*_{t_n}(A)[\Phi(\sigma_{t_n})]^{1/2}=\Phi\left([\sigma_{t_n}]^{1/2}A[\sigma_{t_n}]^{1/2}\right)
$$
and since
$\,\lim_{n\rightarrow+\infty}[\sigma_{t_n}]^{1/2}A[\sigma_{t_n}]^{1/2}=[\sigma]^{1/2}A[\sigma]^{1/2}\,$
in the trace norm topology, the above limit coincides with
$\,\Phi(\left[\sigma]^{1/2}A[\sigma]^{1/2}\right)\,$. So, we have
$\Theta^*_{0}(A)=\Theta^*_{\sigma}(A)$ for all $A\in\B(\H_A)$ and hence
$\Theta_0=\Theta_{\sigma}$.

The above observation shows that for an arbitrary sequence $\{t_n\}$
converging to zero any partial limit of the sequence
$\{\Theta_{t_n}\}$ coincides with $\Theta_{\sigma}$, which means
(\ref{s-lim-one}).\smallskip

\begin{lemma}\label{l-one}
\emph{Let $\,\{\rho_n\}$ be a sequence of states in $\,\S(\H)$
converging to a state $\rho_0$ and $\,\{A_n\}$ a sequence of
operators in the unit ball of $\,\B(\H)$ converging to an operator
$A_0$ in the weak operator topology. Then the sequence
$\,\{\sqrt{\rho_n}A_n\sqrt{\rho_n}\}$ converges to the operator
$\sqrt{\rho_0}A_0\sqrt{\rho_0}$ in the Hilbert-Schmidt norm
topology.}
\end{lemma}\smallskip

\textbf{Proof.} Since $\{\rho_n\}_{n\geq0}$ is a compact set, the
compactness criterion for subsets of $\S(\H)$ (see \cite[Proposition
in the Appendix]{H-Sh-2}) implies that for an arbitrary
$\varepsilon>0$ there exists a finite rank projector
$P_{\varepsilon}$ such that $\Tr
\bar{P}_{\varepsilon}\rho_n<\varepsilon$ for all $n\geq0$, where
$\bar{P}_{\varepsilon}=I_{\H}-P_{\varepsilon}$. We have
\begin{equation}\label{a-e}
\begin{array}{c}
\sqrt{\rho_n}A_n\sqrt{\rho_n}=\sqrt{\rho_n}P_{\varepsilon}A_n
P_{\varepsilon}\sqrt{\rho_n}\\\\+\sqrt{\rho_n}P_{\varepsilon}A_n
\bar{P}_{\varepsilon}\sqrt{\rho_n}+\sqrt{\rho_n}\bar{P}_{\varepsilon}A_n
P_{\varepsilon}\sqrt{\rho_n}+\sqrt{\rho_n}\bar{P}_{\varepsilon}A_n
\bar{P}_{\varepsilon}\sqrt{\rho_n},\quad n\geq 0,
\end{array}
\end{equation}
Since $P_{\varepsilon}$ has finite rank, $P_{\varepsilon}A_n
P_{\varepsilon}$ tends to $P_{\varepsilon}A_0 P_{\varepsilon}$ in
the norm topology and hence $\sqrt{\rho_n}P_{\varepsilon}A_n
P_{\varepsilon}\sqrt{\rho_n}$ tends to
$\sqrt{\rho_0}P_{\varepsilon}A_0 P_{\varepsilon}\sqrt{\rho_0}$ the
trace norm topology, while it is easy to show that the
Hilbert-Schmidt norm of the other terms in the right-hand side of
(\ref{a-e}) tends to zero as $\,\varepsilon\rightarrow0$ uniformly
on $n$. $\square$

\subsection{Some auxiliary results}

\begin{lemma}\label{app}
\emph{An arbitrary complete orthogonally non-decomposable family of
pure states in a separable Hilbert space $\H$ contains a countable complete orthogonally
non-decomposable subfamily.}
\end{lemma}
\smallskip

\textbf{Proof.} Let $\mathfrak{H}$ be the set of all subspaces of
$\H$ generated by countable OND
subfamilies of the family $\S$ endowed with the inclusion ordering. Let
$\mathfrak{H}_0$ be a chain in $\mathfrak{H}$ and
$\H_0=\overline{\bigcup_{\K\in\mathfrak{H}_0}\K}$. Since there is a
countable chain $\{\H_k\}$ in $\mathfrak{H}$ such that
$\H_0=\overline{\bigcup_{k}\H_k}$ and a countable union of countable OND subfamilies is
a countable OND subfamily, the subspace
$\H_0$ belongs to the set $\mathfrak{H}$.  Hence $\H_0$ is an upper
bound of the chain $\mathfrak{H}_0$ and Zorn's lemma implies
existence of a maximal element $\H_m$ in $\mathfrak{H}$. Suppose,
$\H_m\varsubsetneq\H$. Since the family $\S$ is complete and
orthogonally non-decomposable, it contains a pure state
$|\varphi\rangle\langle\varphi|$ such that the
vector $|\varphi\rangle$ lies neither in $\H_m$ nor in
$\H_m^{\perp}$. By adding the state $|\varphi\rangle\langle\varphi|$
to the countable OND subfamily
corresponding to the subspace $\H_m$ we obtain a countable
OND subfamily. Hence
$\H_m\vee\{\lambda|\varphi\rangle\}\in\mathfrak{H}$ contradicting to
the maximality of $\H_m$ $\square$. \medskip

\begin{lemma}\label{onb}
\emph{Let $\{|\varphi_i\rangle\}$ be a basis in a Hilbert space $\H$
(in the sense that an arbitrary vector $|\psi\rangle$ in $\H$ has a unique
decomposition $|\psi\rangle=\sum_i c_i|\varphi_i\rangle$). Then the
set $\,\{|\phi_i\rangle\}$ of vectors defined in (\ref{v-phi-rep})
by means of an arbitrary non-degenerate probability distribution
$\,\{\pi_i\}$ is an orthonormal basis in $\H$.}
\end{lemma}
\smallskip

\textbf{Proof.} Since $\sum_i|\phi_i\rangle\langle\phi_i|=I_{\H}$,
for given arbitrary $j$ we have
$$
|\phi_j\rangle=\sum_i \langle\phi_i|\phi_j\rangle |\phi_i\rangle
$$
and hence
$$
(\|\phi_j\|^2-1)|\phi_j\rangle+\sum_{i\neq j} \langle\phi_i|\phi_j\rangle
|\phi_i\rangle=0.
$$
By applying the
operator $\bar{\rho}_{\pi}=\sum_i\pi_i|\varphi_i\rangle\langle\varphi_i|\,$ to all the terms of this vector equality  we obtain
$$
\sqrt{\pi_j}(\|\phi_j\|^2-1)|\varphi_j\rangle+\sum_{i\neq j}
\sqrt{\pi_i}\langle\phi_i|\phi_j\rangle |\varphi_i\rangle=0.
$$
Since $\{|\varphi_i\rangle\}$ is a basis and $\pi_i>0$ for all $i$, we have $\|\phi_j\|^2=1$
and $\langle\phi_i|\phi_j\rangle=0$ for all $i\neq j$. Thus
$\{|\phi_i\rangle\}$ is an orthonormal system of vectors in $\H$. It
is a complete system, since $\sum_i|\phi_i\rangle\langle\phi_i|=I_{\H}$.
$\square$
\bigskip

I am grateful to A.S.Holevo and to the participants of his seminar "Quantum probability, statistic, information"
(the Steklov Mathematical Institute) for the
useful discussion. I am also grateful to A.Jencova and to T.Shulman for
the valuable help in solving the particular questions.

The work is supported in part by the
Scientific Program "Mathematical Control Theory and Dynamic Systems" of the Russian
Academy of Sciences and the Russian Foundation for Basic Research, projects
10-01-00139-a and 12-01-00319-a.

\end{document}